\documentclass[reqno,11pt]{amsart}
\usepackage{amsmath,amsfonts,amsgen,amstext,amsbsy,amsopn,amsthm, amssymb}

\newtheorem{Lemma}{Lemma} \newtheorem{Theorem}{THEOREM}
\newtheorem{prop}{Proposition} 
 \theoremstyle{definition}

\newtheorem{remark}{Remark}

\newcommand\nn\nonumber

 \newcommand{\R}{\mathbb{R}}
 \newcommand{\N}{\mathbb{N}}
 
\newcommand{\Aa}{\mathcal{A}}
\newcommand{\Cc}{\mathcal{C}}

\renewcommand\rho\varrho
\renewcommand\epsilon\varepsilon

 \DeclareMathOperator{\const}{const}
 \DeclareMathOperator{\re}{Re}

\begin{document}

\title[Lieb-Thirring Inequality for Particles with Point
  Interactions]{Lieb-Thirring Inequality for a Model of Particles with Point
  Interactions}

\author[R.L. Frank]{Rupert L. Frank} 
\address{Department of Mathematics, Princeton University, Princeton, NJ 08544, USA}
\email{rlfrank@math.princeton.edu}

\author[R. Seiringer]{Robert Seiringer}
\address{Department of Mathematics and Statistics, McGill
  University, 805 Sherbrooke Street West, Montreal, QC H3A 2K6,
  Canada} 
\email{robert.seiringer@mcgill.ca}

\date{December 23, 2011}

\begin{abstract}
  We consider a model of quantum-mechanical particles interacting via
  point interactions of infinite scattering length. In the case of
  fermions we prove a Lieb-Thirring inequality for the energy, i.e.,
  we show that the energy is bounded from below by a constant times
  the integral of the particle density to the power $\frac 53$.
\end{abstract}

\thanks{\copyright\, 2011 by  
  the authors. This paper may be reproduced, in its entirety, for
  non-commercial purposes.\\ Submitted to Journal of Mathematical Physics for the special issue celebrating Elliott Lieb's $80^{\rm th}$ birthday.}

\maketitle

\centerline{
{\it To Elliott Lieb, in appreciation of many years of fruitful and inspiring collaboration.}}
\bigskip\medskip

\section{Introduction}

Quantum-mechanical models of particles with point interactions have
become relevant in cold-atom physics, where one encounters systems
where the range of the interactions among the atoms can be much
shorter than their average distance, while the scattering length can
be much larger. In the limit of zero interaction range and infinite
scattering length, which is referred to as the \lq\lq unitary limit\rq\rq, one
is left with a system without intrinsic length scale. For further
discussion of this topic we refer to \cite{JHP,BH,BPST,WC,GW} and, in
particular, to the articles in \cite{zwerger}.

While the two-body problem for particles interacting via point
interactions is well understood \cite{albb}, it remains an open
problem to establish the existence of a model for $N>2$ particles with
only two-body point interactions (see \cite{DFT,TF} and references
there). Such a model can only be expected to exist for spin $\frac 12$
fermions, due to the Efimov effect \cite{efimov}, i.e., the existence
of three-body bound states despite the absence of two-body bound
states. We consider here a model with more complicated point interactions
which is manifestly well-defined, in the sense the quadratic form for
the energy is positive. This model is well-defined even for bosons, and
does not allow for any bound states even in this case.

We shall prove that the model we consider satisfies a Lieb-Thirring inequality,
i.e., the energy of fermions is bounded from below a
semiclassical approximation to the kinetic energy. Up to the value of
the constant, this inequality is the same as the one obtained by Lieb
and Thirring \cite{LT1,LT2} for the kinetic energy only, i.e., in the
absence of any interaction. We recall that this inequality of Lieb and
Thirring played a crucial role in a short and elegant proof of the
stability of matter \cite{stabmatt}, which had been proved earlier by
Dyson and Lenard \cite{DL} by different means.

Since our model concerns interacting particles, we cannot use the
method of Lieb and Thirring to reduce the problem to the spectral
analysis of a one-body operator. Instead, our strategy is closer in
spirit to the work of Dyson and Lenard \cite{DL}, where the Pauli
principle enters only via a local exclusion principle, which implies
that the local kinetic energy cannot be zero for more than one
particle. A similar strategy was recently employed in \cite{LuSo} to
study a system of anyons in two dimensions, and ideas from \cite{LuSo}
are crucial to our proof.

\section{Model and Main Result}

Let $N\geq 2$, $X=(x_1,\dots,x_N)\in \R^{3N}$ and let $g:\R^{3N}\to \R$ denote the function 
\begin{equation}\label{def:g}
g(X) = \sum_{1\leq i<j\leq N} \frac 1{|x_i-x_j|} \,.
\end{equation}
Our model is defined via the quadratic form
\begin{equation}\label{def:Q}
Q(f) = \sum_{i=1}^N  \int_{\R^{3N}} g(X)^2 \left| \nabla_i f(X) \right|^2 dX
\end{equation}
on $L^2(\R^{3N},g(X)^2 dX)$, where $\nabla_i$ stands for the gradient
with respect to $x_i\in \R^3$, and $dX = \prod_{k=1}^N dx_k$. The norm
on the space $L^2(\R^{3N},g(X)^2 dX)$ will simply be denoted by $\|\,\cdot\,\|$. The model
(\ref{def:Q}) was introduced in \cite{alb}, where it was shown that
the quadratic form $Q(f)$ gives rise to a non-negative
self-adjoint operator with purely absolutely continuous spectrum
$[0,\infty)$.

We denote by $\Aa_q^N \subset L^2(\R^{3N},g(X)^2 dX)$ those functions
$f$ that have the property that there is a partition of
$\{1,\dots,N\}$ into $q$ disjoint subsets such that $f$ is
antisymmetric in the variables corresponding to each subset. In
particular, $\Aa_1^N$ are the totally antisymmetric functions, while
$\Aa_N^N$ are all functions in $L^2(\R^{3N},g(X)^2 dX)$. Note that
$\Aa_q^N\subset \Aa_{q+1}^N$. The functions $f\in \Aa_q^N$ thus
represent the spatial part of totally anti-symmetric functions of
space and an internal degree of freedom (\lq\lq spin\rq\rq), where the spin is
allowed to take $q$ different values.

For given $f \in L^2(\R^{3N},g(X)^2 dX)$ we define its density $\rho_f\in L^1(\R^3)$ by 
\begin{equation}\label{def:rho}
\rho_f(x) = \sum_{i=1}^N \int_{\R^{3(N-1)}} g(\hat X_i)^2 |f(\hat X_i)|^2 d\hat X_i
\end{equation}
with $\hat X_i=(x_1,\dots,x_{i-1},x,x_{i+1},\dots,x_N)$ and $d\hat
X_i = \prod_{j\neq i} dx_j$. Our main result is the following
Lieb-Thirring type inequality.

\begin{Theorem}\label{thm1}
 For some constant $C>0$ (independent of $N$ and $q$) we have 
\begin{equation}\label{met}
Q(f) \geq \frac C{q^{2/3}} \int_{\R^3} \rho_f(x)^{5/3} dx 
\end{equation}
for all $f\in H^1(\R^{3N},g(X)^2 dX) \cap \Aa_q^N$ with $\|f\|=1$.
\end{Theorem}

As mentioned in the introduction, Lieb and Thirring proved
(\ref{met}) in the case $g\equiv 1$. Improved bounds on the constant
$C$ in this case were obtained in \cite{DLL} (see also \cite{LW,Hu}).

For $\Omega\subset \R^3$ open and with finite measure, the ground
state energy for $N$ particles confined to $\Omega$ equals
\begin{equation}
E(N,q,\Omega) = \inf \left\{ Q(f) \, : \, f \in C_0^\infty(\Omega^N   ) \cap \Aa_q^N ,\ \|f\|=1 \right\}
\end{equation}
for our model. 
An immediate corollary of (\ref{met}) and H\"older's inequality is that
\begin{equation}
E(N,q,\Omega) \geq \frac C{q^{2/3}} \frac{N^{5/3}}{|\Omega|^{2/3}}\,,
\end{equation}
where $|\Omega|$ denotes the volume of $\Omega$.

\begin{remark}
With $\Delta_k$ denoting the Laplacian with respect to the variables $x_k$, we have $\sum_{k=1}^N  \Delta_k \, g(X)=0$ for all $X\in \R^{3N}$ with $|x_i-x_j|>0$ for all $1\leq i<j\leq N$.
If we define
\begin{equation}
T_\epsilon= \{ X \in \R^{3N} \, :\, |x_i-x_j|>\epsilon \ \text{$\forall 1\leq i<j\leq N$}\}
\end{equation}
for $\epsilon>0$, then an integration by parts gives
\begin{align}\nonumber
\sum_{i=1}^N  \int_{T_\epsilon} g(X)^2 \left| \nabla_i f(X) \right|^2 dX & = \sum_{i=1}^N\int_{T_\epsilon}  \left| \nabla_i g(X) f(X) \right|^2 dX \\ & \quad - 2 \int_{\partial T_\epsilon} |f(X)|^2 g(X) \sum_{1\leq j<k\leq N} \frac{1}{|x_j-x_k|^2} dS \,, \label{intp}
\end{align}
where $dS$ denotes the induced surface measure on $\partial
T_\epsilon$, the boundary of $T_\epsilon$. Hence our model indeed
describes particles with point interactions supported on the union of
the hyperplanes defined by $x_i=x_j$ for $1\leq i<j\leq N$.

Note that the last term in (\ref{intp}) is negative and vanishes in
the limit $\epsilon\to 0$ if $|f(X)|^2$ vanishes faster then linearly
on the hyperplanes where $x_i=x_j$. In particular, if $f(X) =
\Psi(X)/g(X)$ for $\Psi\in C_0^\infty(\R^{3N})$, then
\begin{equation}
Q(f) = \sum_{i=1}^N  \int_{\R^{3N}}  \left| \nabla_i \Psi(X) \right|^2 dX \,.
\end{equation}
In general, the functions $f$ need not vanish for coinciding
arguments, however, and the energy can be lowered due to the
attractive nature of the last term in (\ref{intp}).
\end{remark}

\begin{remark}
  If we replace $1/|x|$ in (\ref{def:g}) by a smooth, strictly
  positive functions $\varphi(x)$ and define, accordingly, $\widetilde
  g(x) = \sum_{1\leq i<j\leq N} \varphi(x_i-x_j)$ and $\Psi(X) =
  \widetilde g(X) f(X)$, then an integration by parts gives
\begin{align}\nonumber
  \sum_{i=1}^N \int_{\R^{3N}} \widetilde g(X)^2 \left| \nabla_i f(X)
  \right|^2 dX & = \sum_{i=1}^N\int_{\R^{3N}} \left| \nabla_i \Psi(X)
  \right|^2 dX \\ & \quad + \int_{\R^{3N}} | \Psi(X)|^2 \frac{
    \sum_{k=1}^N \Delta_k \widetilde g(X)}{\widetilde g(X)} dX
  \,. \label{intpa}
\end{align}
The effective potential $\widetilde g(X)^{-1} \sum_{k=1}^N \Delta_k
\widetilde g(X)$ can alternatively be written as
\begin{equation}\label{sumrep}
2 \sum_{1\leq i<j\leq N}  \frac{ \Delta \varphi(x_i-x_j)}{\widetilde g(X)}
\end{equation}
and thus contains not only two-body terms but terms involving arbitrarily
many particles. In particular, the presence of other particles besides
$i$ and $j$ decreases the summands in (\ref{sumrep}) and hence weakens
the interaction. This weakening leads to improved stability properties
as compared to the case with only two-body interaction.
\end{remark}

We note that our main result Theorem~\ref{thm1} holds for more general
functions $g(X)$. As an example, one could replace $g(X)$ by
\begin{equation}
\sum_{1\leq i<j\leq N} \left( \frac 1{|x_i-x_j|} - \frac 1 a\right)
\end{equation}
for $a<0$, corresponding to point interactions with finite scattering length
$a$. A Lieb-Thirring inequality holds also in this case, with a
constant $C$ that is bounded from below uniformly in $a$. As $a\to 0$
one obtains non-interacting particles. 

Our method can also be applied to analogous models in dimensions $n>3$,
where the function $1/|x|$ in (\ref{def:g}) should be replaced by
$|x|^{2-n}$, the Green's function of the Laplacian in $\R^n$.

The remainder of this paper is concerned with the proof of
Theorem~\ref{thm1}. The proof will be split into two main parts. In
the first part in Section~\ref{sec:lep}, we prove a local exclusion
principle (Prop.~\ref{fL}) which states that the energy of $n$
particles in a finite cube is strictly positive for $n>q$ and grows at least 
like $n-q$, in fact. The second part in Section~\ref{sec:h2o} concerns
the proof that the energy $Q(f)$ dominates the $L^2(\R^3)$ norm of
$\nabla \sqrt{\rho_f}$. The inequality we prove in
Prop.~\ref{lemma:db} is, up to a constant, equal to a well-known
inequality by the Hoffmann-Ostenhofs \cite{h2o} in the case without
interaction. In Section~\ref{sec:final} we demonstrate how to obtain
Theorem~\ref{thm1} from Propositions~\ref{fL} and~\ref{lemma:db}. The
strategy of the proof is  similar to the one in \cite{LuSo} where
a Lieb-Thirring inequality was proved for anyons in two dimensions.

\section{Local Exclusion Principle}\label{sec:lep}

In this section we shall prove the following.

\begin{prop}\label{fL}
Let $f \in H^1(\R^{3N},g(X)^2dX) \cap \Aa_q^N$ with $\|f\|=1$, and let $\Cc_L\subset \R^3$ be a cube of side length $L>0$. Then 
\begin{equation}\label{ers}
\sum_{i=1}^N  \int_{\R^{3N}} g(X)^2 \left| \nabla_i f(X) \right|^2 \chi_{\Cc_L}(x_i) dX \geq \frac k{L^2} \left( \int_{\Cc_L} \rho_f(x) dx  - q \right)
\end{equation}
for a constant $k>0$ independent of $N$, $f$, $L$ and the location of the cube $\Cc_L$. 
\end{prop}

Here and in the following, $\chi_{Q}$ denotes the characteristic function of a set $Q\subset \R^3$. 
The constant $k$ appearing in (\ref{ers}) is the same as the one in Lemma~\ref{pL} below.

We note that this proposition implies that, for any integer $M\geq 1$, 
\begin{equation}\label{te}
E(N,q,\Cc_L) \geq k  \frac{ N M^2  - q M^5}{L^2} \,.
\end{equation}
To prove (\ref{te}),  simply divide the cube $\Cc_L$ into $M^3$ disjoint cubes with side length $L/M$ and apply Proposition~\ref{fL} for every cube.
The optimal choice of $M$ is close to $(5N/2q)^{1/3}$. In particular, the maximum over $M$ of the right side of (\ref{te}) is proportional to  $k N^{5/3}/(L^2q^{2/3})$ for large $N$.

The proof of Proposition~\ref{fL} will be divided into several lemmas. 
The following lemma is well-known \cite{cw}; we include its proof for the convenience of the reader.

\begin{Lemma} Let $\Omega$ be an open and convex subset of $\R^n$ with diameter $d<\infty$, and let $\omega$ be a strictly positive function in $L^1(\Omega)$. For all $f\in H^1(\Omega,\omega(x)dx)$ with $\int_\Omega f(x) dx = 0$, we have
\begin{equation}
\int_\Omega |\nabla f(x)|^2 \omega(x) dx \geq \frac 1{2^3 3^n |\mathbb{B}^n|^2 M_\omega^3} \frac{|\Omega|^2}{d^{2(n+1)}} \int_\Omega |f(x)|^2 \omega(x) dx \,,
\end{equation}
where $|\mathbb{B}^n|$ denotes the volume of the unit ball in $\R^n$ and 
\begin{equation}\label{def:mo}
M_\omega = \sup_{r>0,x\in\Omega} \frac 1{|\Omega\cap B_r(x)|} \frac{ \int_{\Omega\cap B_r(x)}  \omega(z) dz}{\inf_{z\in \Omega\cap B_r(x)} \omega(z)}
\end{equation}
with $B_r(x)$ denoting the ball of radius $r$ centered at $x$.
\end{Lemma} 

\begin{proof}
Since $\Omega$ is convex, we can write
\begin{equation}
f(x) - f(y) = \int_0^1 (x-y)\nabla f(t x + (1-t) y) dt \,.
\end{equation}
Integrating this identity over $y\in \Omega$ and changing variables to $z=tx + (1-t)y$ gives
\begin{equation}
f(x) = \frac 1{|\Omega|} \int_0^1 \frac 1{(1-t)^{n+1}} \int_{\Omega} (x-z) \nabla f(z) \chi_\Omega\left(\frac{z-tx}{1-t}\right) dz\,dt \,.
\end{equation}
Since $|x-z| = (1-t)|x-y|\leq (1-t) d$ for $x,y\in \Omega$, we obtain the bound
\begin{align}\nonumber
|f(x)| & \leq  \frac 1{|\Omega|}\int_{\Omega} |x-z| | \nabla f(z) | \int_0^{1-|x-z|/d} \frac 1{(1-t)^{n+1}}\,  dt\, dz \\ & \leq \frac {d^n}{n |\Omega|}\int_{\Omega} \frac 1{|x-z|^{n-1}} |\nabla f(z)| dz\,. \label{19}
\end{align}
For $x,z\in\Omega$, we have 
\begin{equation}
\frac 1{|x-z|^{n-1}} = (n-1) \int_0^d \frac 1{r^n} \chi_{B_r(x)}(z) dr + \frac 1{d^{n-1}}\,.
\end{equation}
Hence (\ref{19}) can be rewritten as 
\begin{equation}\label{21}
|f(x)| \leq  \frac {d^n}{n |\Omega|}\left[  (n-1) \int_0^d \frac 1 {r^n} \int_{\Omega\cap B_r(x)} |\nabla f(z)| dz \, dr +  \frac 1{d^{n-1}} \int_\Omega |\nabla f(z)| dz\right]\,.
\end{equation}

We introduce the maximal function
\begin{equation}
m_\omega(x) = \sup_{r>0} \frac{  \int_{\Omega\cap B_r(x)} |\nabla f(z)| \omega(z) dz}{\int_{\Omega\cap B_r(x)}  \omega(z) dz} \,.
\end{equation}
Then
\begin{equation}
\int_{\Omega\cap B_r(x)} |\nabla f(z)| dz \leq |\mathbb{B}^n| r^n M_\omega m_\omega(x)
\end{equation}
with $M_\omega$ defined in (\ref{def:mo}). 
In particular, from (\ref{21}) we obtain the bound 
\begin{equation}
|f(x)|\leq  \frac {d^{n+1}}{ |\Omega|} |\mathbb{B}^n| M_\omega m_\omega(x)\,,
\end{equation}
and thus
\begin{equation}\label{25}
\int_\Omega |f(x)|^2 \omega(x) dx \leq \frac{d^{2(n+1)}}{|\Omega|^2} |\mathbb{B}^n|^2 M_\omega^2 \int_\Omega m_\omega(x)^2 \omega(x) dx\,.
\end{equation}

We further define
\begin{equation}
N_\omega = \sup_{r>0,x\in\Omega} \frac{ \int_{\Omega\cap B_{3r}(x)}  \omega(z) dz}{\int_{\Omega\cap B_r(x)}  \omega(z) dz }
\end{equation}
and note that $N_\omega \leq  3^n M_\omega$. 
Proceeding as in the proof of Theorem 1(c) in Sect. I.3.1 of \cite{stein} we see that 
\begin{equation}
\int_\Omega m_\omega(x)^2 \omega(x) dx \leq 2^3 N_\omega \int_\Omega |\nabla f(x)|^2 \omega(x) dx \,.
\end{equation}
In combination with (\ref{25}) this  proves our claim.
\end{proof}

Let now $\Cc = (0,1)^3 \subset \R^3$ denote the unit cube in $\R^3$.

\begin{Lemma}\label{pL} Assume that $f \in H^1(\R^3\times \R^3)$ is
  antisymmetric, i.e., $f(x_1,x_2) = -f(x_2,x_1)$ for $x_1,x_2\in
  \R^3$. For almost every $(x_3,\dots,x_N)\in \R^{3(N-2)}$ we have
\begin{equation}\label{il}
\sum_{i=1}^2   \int_{\Cc^{2}} g(X)^2 \left| \nabla_i f(x_1,x_2) \right|^2 
 dx_1dx_2 \\ \geq 2 k  \int_{\Cc^{2}} g(X)^2 |f(x_1,x_2)|^2 dx_1 dx_2
\end{equation}
for a constant $k>0$ that does not depend on $N$ or the variables $(x_3,\dots,x_N)$.
\end{Lemma} 

\begin{proof}
  Since $f$ is antisymmetric, it has zero average with respect to the
  variables $(x_1,x_2)$. Hence we can apply the previous lemma, with
  $\Omega = \Cc \times \Cc \subset \R^6$ and $\omega(x_1,x_2) =
  g(X)^2$. The claim is thus proved if we can show that the
  corresponding $M_\omega$ in (\ref{def:mo}) is bounded independently
  of $N$ and $(x_3,\dots,x_N)$.

  Consider a ball $B \subset \R^6$ of radius $r$ centered at some
  point $(w_1,w_2)\in \Cc\times \Cc$. Without loss of generality, we
  may assume that $r\leq \sqrt{6}$, for otherwise $\Cc\times \Cc \subset B$
  and hence the expression after the $\sup$ in (\ref{def:mo}) is independent
  of $r$, and is thus the same as for $r=\sqrt{6}$. We can get a lower bound on
  $g(X)$ by taking the minimum in each term in the sum in
  (\ref{def:g}). This gives
\begin{equation}
g(X) \geq \frac{1}{|w_1-w_2|+2r} + \sum_{j\geq 3}\left( \frac{1}{|w_1-x_j|+r} + \frac 1{|w_2-x_j|+r} \right) + \sum_{3\leq j<l\leq N} \frac 1{|x_j-x_l|}
\end{equation}
for $(x_1,x_2)\in B$. On the other hand, simple estimates yield
\begin{equation}
\int_B \frac{1}{|x_1-x_2|^2} dx_1 dx_2 \leq \const \frac {r^6} {(|w_1-w_2|+2r)^2} 
\end{equation}
and 
\begin{equation}
\int_B \frac{1}{|x_1-x_j||x_1-x_l|} dx_1 dx_2 \leq \const \frac {r^6} {(|w_1-x_j|+2r)(|w_1-x_l|+2r)} 
\end{equation}
for $j,l\geq 3$, etc. These imply, in particular, that 
\begin{equation}
\int_B g(X)^2 dx_1 dx_2 \leq \const r^6 \inf_{(x_1,x_2)\in B} g(X) \,.
\end{equation}
Since $\Omega$ is a cube in $\R^6$, we also have $|\Omega\cap B| \geq
2^{-6} r^6$ for $r\leq 1$, and hence $M_\omega$ is bounded
independently of $N$ and $(x_3,\dots,x_N)$.
\end{proof}

\begin{Lemma} Let $A\subset \{1,\dots,N\}$, $X_A = \{x_i\}_{i\in A}$ and $d X_A =\prod_{l\in A} dx_l$. For $f \in H^1(\R^{3N},g(X)^2 dX) \cap \Aa_q^N$  we have, for almost every $X\setminus X_A$, 
\begin{equation}\label{il2}
\sum_{i \in A}  \int_{\Cc^{|A|}} g(X)^2 \left| \nabla_i f(X) \right|^2 
 dX_A \\ \geq k  \left( |A|-q\right) \int_{\Cc^{|A|}} g(X)^2 |f(X)|^2 
   dX_A
\end{equation}
with $k$ as in Lemma~\ref{pL}.
\end{Lemma} 

\begin{proof} The set $A$ can be divided into $q$ subsets $B_j$ (some
  of which may be empty) such that $f$ is antisymmetric in the
  variables in each subset. We may assume that $|A|>q$, in which case
  at least one such subset $B_j$ contains more than one element. We
  shall show that for all such $B_j$
\begin{equation}\label{34}
\sum_{i \in B_j}  \int_{\Cc^{|B_j|}} g(X)^2 \left| \nabla_i f(X) \right|^2 
 dX_{B_j} \\ \geq k \left( |B_j|-1\right) \int_{\Cc^{|B_j|}} g(X)^2 |f(X)|^2 
   dX_{B_j}
\end{equation}
for each fixed $X\setminus X_{B_j}$. Summing over $j$ this implies the
result. 

Inequality  (\ref{34}) follows immediately from Lemma~\ref{pL},
noting that the sum over $|B_j|=n\geq 2$ elements can be written as
$1/[2(n-1)]$ times a sum over all ordered pairs, yielding the inequality with
$|B_j|-1$ replaced by $|B_j|$.
\end{proof}

Since $g(X)$ is invariant under translation and rotation of all the
coordinates $x_i$, inequality (\ref{il2}) clearly holds for all unit
cubes in $\R^3$, irrespective of their location or orientation. A
simple scaling argument shows that it actually holds for all cubes of
side length $L>0$ if the right side is divided by $L^2$.

\begin{proof}[Proof of Proposition~\ref{fL}]
Using that 
\begin{equation}
1 = \sum_{A\subset\{1,\dots,N\}} \prod_{l\in A} \chi_{\Cc_L}(x_l) \prod_{l\not\in A} \left( 1 - \chi_{\Cc_L}(x_l)\right)
\end{equation}
we have 
\begin{multline}
\sum_{i=1}^N  \int_{\R^{3N}} g(X)^2 \left| \nabla_i f(X) \right|^2 \chi_{\Cc_L}(x_i) dX \\ = \sum_{A\subset\{1,\dots,N\}} \sum_{i \in A}  \int_{\R^{3N}} g(X)^2 \left| \nabla_i f(X) \right|^2 
 \prod_{l\in A} \chi_{\Cc_L}(x_l) \prod_{l\not\in A} \left( 1 - \chi_{\Cc_L}(x_l)\right) dX \,.
\end{multline}
Applying the previous lemma (and the remark after its proof), we obtain the lower bound
\begin{equation}
\frac {k}{L^2}  \sum_{A\subset\{1,\dots,N\}} \left(|A|-q\right)  \int_{\R^{3N}} g(X)^2 \left| f(X) \right|^2 
 \prod_{l\in A} \chi_{\Cc_L}(x_l) \prod_{l\not\in A} \left( 1 - \chi_{\Cc_L}(x_l)\right) dX  \,,
\end{equation}
which is equal to the right side of (\ref{ers}).
\end{proof}

\section{Density Bound}\label{sec:h2o}

Recall that the one-particle density $\rho_f$ of a function $f\in
L^2(\R^{3N},g(X)^2 dX)$ is defined in (\ref{def:rho}).

\begin{prop}\label{lemma:db}
Let $f\in H^1(\R^{3N},g(X)^2 dX)$. Then its one-particle density $\rho_f$ satisfies $\sqrt{\rho_f}\in H^1(\R^3,dx)$ and
\begin{equation}\label{38}
Q(f) \geq \frac 1 9 \int_{\R^3} \Big| \nabla \sqrt{\rho_f(x)} \Big|^2 dx \,.
\end{equation}
\end{prop}

It is well known that in case $g\equiv 1$, (\ref{38}) holds with the prefactor $\frac 19$ replaced by $1$ \cite{h2o}.

\begin{proof}
  We first show that it is enough to prove the inequality for
  symmetric functions $f$. For any $f$, define $\tilde f$ by
\begin{equation}
\tilde f(X) = \left( \frac 1{N!} \sum_{\pi \in S_N} | f(x_{\pi(1)},\dots,x_{\pi(N)}) |^2 \right)^{1/2}\,,
\end{equation}
where $S_N$ denotes the permutation group. The function $\tilde f$ is obviously symmetric, and $\rho_f= \rho_{\tilde f}$. It is easy to see that the map $|f|^2 \mapsto Q(|f|)$ is convex (see, e.g., Lemma~A.1 in \cite{lsy}), hence
\begin{equation}
Q(\tilde f) \leq Q(|f|) \leq Q(f) \,.
\end{equation}
In particular, it is enough to prove (\ref{38}) for symmetric $f$.

Assume now that $f$ is symmetric. Then the density $\rho_f$ can alternatively be written as 
\begin{align}\nonumber
\rho_f(x_1) & = N \int_{\R^{3(N-1)}} g(X)^2 |f(X)|^2 dx_2\cdots dx_N \\ \nonumber & = N(N-1) \int_{\R^{3(N-1)}} \frac{1}{|x_1-x_2|^2} |f(X)|^2 dx_2\cdots dx_N \\ \nonumber & \quad + N(N-1)(N-2) \int_{\R^{3(N-1)}} \frac{1}{|x_1-x_2|} \frac 1{|x_1-x_3|} |f(X)|^2 dx_2\cdots dx_N \\ \nonumber & \quad + 2N(N-1)(N-2) \int_{\R^{3(N-1)}} \frac 1{|x_1-x_2|} \frac 1{|x_2-x_3|} |f(X)|^2 dx_2\cdots dx_N \\ \nonumber & \quad + N(N-1)(N-2)(N-3) \int_{\R^{3(N-1)}} \frac 1{|x_1-x_2|}\frac 1{|x_3-x_4|} |f(X)|^2 dx_2\cdots dx_N \\ & \quad + N \int_{\R^{3(N-1)}} R(x_2,\dots,x_N) |f(X)|^2 dx_2\cdots dx_N
\end{align}
where $R$ denotes the sum of all the terms in $g$ not depending on
$x_1$. Let us denote the various integrals on the right side by
$\rho_1,\dots,\rho_5$, and the prefactors in front of the integrals by
$n_1,\dots,n_5$. That is, $\rho_f = \sum_{j=1}^5 n_j \rho_j$.

Since $1/|x_1-x_2|$ is invariant under translations of both $x_1$ and $x_2$, we have 
\begin{equation}
(\nabla_1+\nabla_2) \frac 1{|x_1-x_2|^2} |f(X)|^2 = 2 \re \frac 1{|x_1-x_2|^2} f(X)^*(\nabla_1 + \nabla_2) f(X) \,.
\end{equation}
Integrating this identity over $x_2,\dots,x_N$, we obtain
\begin{equation}
\nabla_1 \rho_1(x_1) = 2 \re \int \frac 1{|x_1-x_2|^2} f(X)^* (\nabla_1 + \nabla_2)f(X) dx_2\cdots dx_N \,.
\end{equation}
The Schwarz inequality then implies that
\begin{equation}
|\nabla_1 \rho_1(x_1)|^2 \leq 4 \rho_1(x_1) \int \frac 1{|x_1-x_2|^2} |(\nabla_1 + \nabla_2)f(X)|^2 dx_2\cdots dx_N\,.
\end{equation}
In particular, 
\begin{align}\nonumber 
\int_{\R^3} |\nabla \sqrt {\rho_1}|^2  & = \frac 14 \int_{\R^3} \frac {1}{\rho_1(x)} \left| \nabla \rho_1(x) \right|^2 dx \\ &\leq  \int_{\R^{3N}} \frac 1{|x_1-x_2|^2} |(\nabla_1+\nabla_2)f|^2 dX \leq 4 \int_{\R^{3N}} \frac{1}{|x_1-x_2|^2}|\nabla_1 f|^2 dX \,, 
\end{align}
where we have again used the symmetry of $f$ in the last step. 

In the same way one proceeds for the remaining parts of the density. The result is  that
\begin{align}\nonumber
\int_{\R^3} |\nabla \sqrt {\rho_2}|^2 & \leq  \int_{\R^{3N}} \frac 1{|x_1-x_2|}\frac 1{|x_1-x_3|} |(\nabla_1+\nabla_2+\nabla_3)f|^2 dX \\ \nonumber & \leq 3 \int_{\R^{3N}} \frac{1}{|x_1-x_2|}\frac1{|x_1-x_3|}|\nabla_1 f|^2 dX \\ & \quad + 6 \int_{\R^{3N}} \frac{1}{|x_1-x_2|}\frac1{|x_1-x_3|}|\nabla_2 f|^2 dX \,,
\end{align}
\begin{align}\nonumber
\int_{\R^3} |\nabla \sqrt {\rho_3}|^2 & \leq  \int_{\R^{3N}} \frac 1{|x_1-x_2|}\frac 1{|x_2-x_3|} |(\nabla_1+\nabla_2+\nabla_3)f|^2 dX \\ \nonumber & \leq 6 \int_{\R^{3N}} \frac{1}{|x_1-x_2|}\frac1{|x_2-x_3|}|\nabla_1 f|^2 dX \\ & \quad + 3 \int_{\R^{3N}} \frac{1}{|x_1-x_2|}\frac1{|x_2-x_3|}|\nabla_2 f|^2 dX \,, 
\end{align}
\begin{align}\nonumber 
\int_{\R^3} |\nabla \sqrt {\rho_4}|^2  & \leq  \int_{\R^{3N}} \frac 1{|x_1-x_2|}\frac 1{|x_3-x_4|} |(\nabla_1+\nabla_2)f|^2 dX  \\ &\leq 4 \int_{\R^{3N}} \frac{1}{|x_1-x_2|}\frac1{|x_3-x_4|}|\nabla_1 f|^2 dX
\end{align}
and
\begin{equation}
\int_{\R^3} |\nabla \sqrt {\rho_5}|^2 \leq  \int_{\R^{3N}} R(x_2,\dots,x_N) |\nabla_1f|^2 dX \,.
\end{equation}
Summing up, using again the symmetry of $f$ and convexity of the map
$\rho\mapsto \int |\nabla\sqrt\rho|^2$, this gives
\begin{equation}
\int_{\R^3} |\nabla \sqrt \rho|^2 \leq \sum_{j=1}^5 n_j \int_{\R^3}|\nabla \sqrt\rho_j|^2 \leq 9 N \int_{\R^{3N}} g(X)^2 |\nabla_1 f|^2 dX = 9 Q(f)\,.
\end{equation}
This completes the proof of our claim.
\end{proof}

\section{Proof of Theorem~\ref{thm1}}\label{sec:final}

We will deduce Theorem~\ref{thm1} from Propositions~\ref{fL} and~\ref{lemma:db}, following a similar strategy as in \cite{LuSo}.

First, we note that it is enough to prove the theorem for $N> 2q$. If $N\leq 2q$, the statement can be easily deduced from  Proposition~\ref{lemma:db} alone. In fact, applying the Sobolev inequality $\|\nabla \varphi\|_2^2 \geq S \|\varphi\|_6^2$ and H\"older's inequality, we have 
\begin{equation}
Q(f) \geq \frac 19 \| \nabla \sqrt{\rho_f} \|_2^2 \geq \frac S9 \|\rho_f\|_3 \geq \frac S9   \|\rho_f\|_{5/3}^{5/3} \|\rho_f\|_1^{-2/3} \geq  \frac S 9 \frac 1{(2 q)^{2/3}} \int_{\R^3} \rho_f(x)^{5/3} dx
\end{equation}
in this case.

Assume now that $N> 2q$. For given $f\in H^1(\R^{3N},g(X)^2 dX) \cap \Aa_q^N$ and $\epsilon>0$, we can choose a cube $Q_0\subset \R^3$ such that 
\begin{equation}
\int_{Q_0} \rho_f(x) dx \geq 2q \quad \text{and} \quad \int_{Q_0} \rho_f(x)^{5/3}dx \geq (1-\epsilon) \int_{\R^3} \rho_f(x)^{5/3} dx\,.
\end{equation}
Suppose that $Q_0$ is divided into finitely many disjoint cubes $Q_i$. 
With the aid of  Proposition~\ref{fL}, we can bound
\begin{equation}\label{a1}
Q(f) \geq \sum_i  \sum_{j=1}^N  \int_{\R^{3N}} g(X)^2 \left| \nabla_j f(X) \right|^2 \chi_{Q_i}(x_j) dX \geq \sum_i  \frac k{|Q_i|^{2/3}} \left[ \int_{Q_i} \rho_f(x) dx  - q \right]_+ \,,
\end{equation}
where $[t]_+ = \max\{t,0\}$ denotes the positive part of a number
$t\in \R$. Note that each summand on the left side is obviously
non-negative, hence we can use the positive part on the right
side. Moreover, from Proposition~\ref{lemma:db}, we obtain the bound
\begin{equation}\label{a2}
Q(f) \geq \sum_i   \frac{1}{9} \int_{Q_i}
  \Big|\nabla \sqrt{\rho_f(x)}\Big|^2 dx \,.
\end{equation}
In combination, (\ref{a1}) and (\ref{a2}) imply that 
\begin{equation}\label{inla}
Q(f) \geq  \sum_i  \left( \frac {\lambda k}{|Q_i|^{2/3}} \left[ \int_{Q_i} \rho_f(x) dx  - q \right]_+  +  \frac{1-\lambda}{9} \int_{Q_i}
  \Big|\nabla \sqrt{\rho_f(x)}\Big|^2 dx \right)
\end{equation}
for any $0\leq \lambda\leq 1$.  

To construct the division of $Q_0$, we proceed as
follows \cite{LuSo}. Divide the cube $Q_0$ into 8 disjoint cubes of half the size. If the
integral of $\rho_f$ over one of the subcubes is less than $2q$, the subcube 
will not be divided further and will be marked $A$. If all subcubes
are marked $A$, then the division is undone and the cube $Q_0$ is
marked $B$. For all the subcubes with integral of $\rho_f$ bigger or
equal to $2q$, we iterate this procedure. At the end, $Q_0$ is thus
covered by finitely many disjoint subcubes of type either $A$ or $B$.

On every subcube $Q_i$ marked $B$, we have $2q\leq \int_{Q_i}
\rho_f(x)dx <16 q$. This implies, in particular, that 
\begin{equation}
 \left[ \int_{Q_i} \rho_f(x) dx - q \right]_+ \geq \frac 12  \int_{Q_i}
    \rho_f(x) dx\,.
\end{equation}
The Poincar{\'e}-Sobolev inequality on a cube
\cite[Thm.~8.12]{LL} yields
\begin{equation}
\int_{Q_i}
  \Big|\nabla \sqrt{\rho_f(x)}\Big|^2 dx \geq \tilde S \left\| \sqrt{\rho_f} - |Q_i|^{-1}\int_{Q_i}\! \sqrt{\rho_f}\, \right\|_{L^6(Q_i)}^2
\end{equation}
for some constant $\widetilde S>0$ independent of the size or location of $Q_i$. 
The triangle inequality in $L^6(Q_i)$ and a simple Schwarz inequality further imply that 
\begin{equation}
\left\| \sqrt{\rho_f} - |Q_i|^{-1}\int_{Q_i}\! \sqrt{\rho_f}\, \right\|_{L^6(Q_i)}^2 \geq \frac 12 \|\rho_f\|_3 - |Q_i|^{-2/3} \int_{Q_i} \rho_f\,.
\end{equation}
An application of H\"older's inequality finally gives
\begin{equation}\label{f3}
\int_{Q_i}
  \Big|\nabla \sqrt{\rho_f(x)}\Big|^2 dx \geq \frac{ \tilde S}{2} \frac{\int_{Q_i} \rho_f^{5/3}}{ \left( \int_{Q_i} \rho_f \right)^{2/3} } - \tilde S |Q_i|^{-2/3} \int_{Q_i} \rho_f\,.
\end{equation}
If we choose $0<\lambda<1$ in such  a way that 
\begin{equation}
\frac {k\lambda}{2} = \frac{(1-\lambda)\tilde S}9\,,
\end{equation}
i.e., $\lambda = (1+ 9k/(2\tilde S))^{-1}$, we conclude that the contribution of a $B$ cube to the energy in (\ref{inla}) is bigger or equal to 
\begin{equation}
\frac{2^{-11/3}}{\frac 2 k + \frac 9 {\tilde S}}   \frac 1{q^{2/3} }  \int_{Q_i} \rho_f(x)^{5/3}dx\,.
\end{equation}

Consider now a cube labeled $A$, where $\int_{Q_i}\rho_f < 2q$. Pick a $\kappa>2$ and assume, for the moment, that 
\begin{equation}
\int_{Q_i} \rho_f^{5/3} > \kappa |Q_i|^{-2/3} \left( \int_{Q_i}\rho_f \right)^{5/3} \,.
\end{equation}
In this case, it follows from (\ref{f3}) that 
\begin{equation}
\int_{Q_i}
  \Big|\nabla \sqrt{\rho_f(x)}\Big|^2 dx \geq \frac{ \tilde S}{2\kappa} \left( \kappa - 2\right)  \frac{\int_{Q_i} \rho_f^{5/3}}{ \left( 2 q \right)^{2/3} } \,.
\end{equation}
Hence, for our choice of $\lambda$,  the contribution from such an $A$ cube to the energy in (\ref{inla}) is at least
\begin{equation}
\left(1-\frac 2 \kappa\right) \frac{2^{-5/3}}{\frac 2 k + \frac 9 {\tilde S}}   \frac 1{q^{2/3} }  \int_{Q_i} \rho_f(x)^{5/3}dx\,.
\end{equation}

We are left with studying those $A$ cubes where 
\begin{equation}\label{saf}
\int_{Q_i} \rho_f^{5/3} \leq \kappa |Q_i|^{-2/3} \left( \int_{Q_i}\rho_f \right)^{5/3} \,.
\end{equation}
We shall show that their contribution to the total integral of
$\rho_f^{5/3}$ is dominated by the contribution of all the $B$
cubes. In order to see this, we note that our decomposition of $Q_0$
into subcubes can be organized in a tree. (Compare with Fig.~3 in
\cite{LuSo}.) Every $A$ cube can be associated with a $B$ cube, namely
if it can be found by going back in the tree from the $B$ cube,
possibly all the way to $Q_0$, and then one step forward. This allows
to divide the $A$ cubes into groups labeled by the $B$ cubes. This
division is not unique, in general, but this is not important; the
only thing that matters is that every $A$ cube can be associated with
a $B$ cube in this way. Pick a $B$ cube (call it $Q_B$) at level
$l\in\N$ of the tree, and let $\Aa(Q_B)$ be the set of all those
associated $A$ cubes that satisfy (\ref{saf}). Since at every level
$1\leq j\leq l$ of the tree there are at most $7$ $A$ cubes in
$\Aa(Q_B)$, we can bound
\begin{equation}
\sum_{Q_i\in \Aa(Q_B)} \int_{Q_i} \rho_f^{5/3} \leq  \frac{ 7 \kappa (2q)^{5/3}}{ |Q_0|^{2/3}}  \sum_{j=1}^l 4^j \leq \frac{ 7 \kappa (2q)^{5/3}}{ 3 |Q_0|^{2/3}} 4^{l+1}\,.
\end{equation}
On the other hand, for the associated $B$ cube $Q_B$ we have 
\begin{equation}
\int_{Q_B} \rho_f^{5/3} \geq \frac{\left(\int_{Q_b} \rho_f \right)^{5/3}}{|Q_B|^{2/3}} \geq \frac{ (2q)^{5/3}}{|Q_B|^{2/3}} = \frac{(2q)^{5/3}}{ |Q_0|^{2/3}} 4^l \,,
\end{equation}
and thus 
\begin{equation}
\sum_{Q_i\in \Aa(Q_B)} \int_{Q_i} \rho_f^{5/3} \leq \frac{28}{3} \kappa \int_{Q_B} \rho_f^{5/3}\,.
\end{equation}

In particular, the contribution to $\int \rho_f^{5/3}$ of all the $A$ cubes satisfying (\ref{saf}) is bounded by $28\kappa/3$ times the contribution of all the $B$ cubes. In other words, the contribution to $\int \rho_f^{5/3}$ of all the $B$ cubes is bounded from below by $(1+28\kappa/3)^{-1}$ times the contribution of {\em both} the $B$ cubes and the $A$ cubes satisfying (\ref{saf}).

After summing over all cubes, we thus get the lower bound
\begin{equation}
Q(f) \geq \frac {\tilde C}{q^{2/3}} \int_{Q_0} \rho_f(x)^{5/3} dx \geq \frac {\tilde C}{q^{2/3}}(1-\epsilon) \int_{\R^3} \rho_f(x)^{5/3} dx\,,
\end{equation}
with
\begin{equation}
\tilde C =  \frac{2^{-11/3}}{\frac 2 k + \frac 9 {\tilde S}} \sup_{\kappa>2} \min\left\{ \frac 1{1+\frac{28}3 \kappa} \, , \,  4 \left(1 - \frac 2\kappa\right) \right\} = \frac{2^{-2/3}}{\frac 2 k + \frac 9 {\tilde S}}\frac{239 - \sqrt{56977}}{48} > 0\,.
\end{equation}
Since $\epsilon>0$ was arbitrary,  this proves (\ref{met}) with a constant
\begin{equation}
C = \min\left\{ \tilde C \, , \, \frac{2^{-2/3}S}9 \right\}\,.
\end{equation}
\qed

\bigskip
\noindent {\it Acknowledgments.} Partial financial support by 
U.S. NSF grant PHY-1068285 (R.F.) and the NSERC (R.S.) is gratefully
acknowledged.


\end{document}